\documentclass[11pt]{article}

\usepackage[margin=1.2in]{geometry}
\usepackage[T1]{fontenc}
\usepackage[dvipsnames]{xcolor}
\usepackage{microtype}
\usepackage{amsmath,amssymb}
\usepackage{mathtools}
\usepackage{amsthm}
\usepackage[numbers,sort&compress]{natbib}

\usepackage[hypertexnames=false,colorlinks=true,linkcolor=BrickRed,citecolor=Green,urlcolor=Blue]{hyperref}
\usepackage[capitalise,noabbrev]{cleveref}
\usepackage{etoolbox}

\theoremstyle{plain}
\newtheorem{theorem}{Theorem}
\newtheorem{lemma}[theorem]{Lemma}
\newtheorem{corollary}[theorem]{Corollary}
\newtheorem{proposition}[theorem]{Proposition}

\theoremstyle{definition}

\theoremstyle{remark}

\crefname{theorem}{Theorem}{Theorems}
\crefname{lemma}{Lemma}{Lemmas}
\crefname{corollary}{Corollary}{Corollaries}
\crefname{proposition}{Proposition}{Propositions}
\crefname{conjecture}{Conjecture}{Conjectures}
\crefname{definition}{Definition}{Definitions}
\crefname{problem}{Problem}{Problems}
\crefname{example}{Example}{Examples}
\crefname{remark}{Remark}{Remarks}
\crefname{appendix}{Appendix}{Appendices}
\apptocmd{\appendix}{%
  \crefalias{section}{appendix}%
}{}{}
\usepackage{soul}
\usepackage{graphicx}
\graphicspath{{../}{./}}
\usepackage{tabularx}
\usepackage{booktabs}
\usepackage{array}
\usepackage{calc}
\usepackage{algorithm}
\usepackage{algpseudocode}
\newenvironment{coflatalgo}{%
  \smallskip\small
  \begin{tabbing}
    \quad\=\qquad\=\qquad\=\qquad\=\qquad\=\qquad\=\qquad\=\qquad\=\qquad\=\qquad\=\qquad\=\qquad\=\qquad\=\kill}%
  {\end{tabbing}\smallskip}
\usepackage{listings}
\providecommand{\tightlist}{\setlength{\itemsep}{0pt}\setlength{\parskip}{0pt}}

\NewDocumentCommand\CSLBibitem{o m}{}
\newlength{\cslhangindent}
\newlength{\csllabelwidth}
\newlength{\cslentryspacingunit}
\title{Derandomizing Karger's Contraction Algorithm for Matroids}
\author{Yu Cong\thanks{\texttt{yucong143@gmail.com}} \and Chao Xu\thanks{\texttt{thechaoxu@gmail.com}. Supported by the National Natural Science Foundation of China under grant 62372093 and by the CCF-Huawei Populus Grove Fund -- Theoretical Computer Science Special Project CCF-HuaweiLK2025003.} \and Yajie Zhao\thanks{\texttt{yajiezhao8@gmail.com}}\\[4pt]{\small University of Electronic Science and Technology of China, Chengdu, China}}
\date{}

\newcommand{\FF}{\mathbb{F}}

\begin{document}

\maketitle
{\let\thefootnote\relax\footnotetext{\textbf{AI Disclosure.} The authors acknowledge the use of Claude Fable 5 to streamline proofs, sharpen constants, and condense the manuscript. All results are due to the authors, who independently verified every argument and take full responsibility for the content.}}

\begin{abstract}
Karger's randomized contraction algorithm finds a minimum-weight cocircuit of a matroid whenever the cogirth-density ratio is bounded. We prove that the same hypothesis yields a deterministic algorithm with the same exponent. If every contraction minor of rank at least \(r_0\) of a matroid \(M\) has cogirth-density ratio at most \(c\), then a minimum-weight cocircuit of \(M\) is computable deterministically in \(m^{O(r_0)}\,n^{O(c)}\) time when the contraction minors of bounded rank have at most \(m\) parallel classes, by an algorithm that knows neither \(r_0\) nor \(c\). As a consequence, we give a deterministic algorithm computing the cogirth of rank-\(p\) perturbed graphic matroids in \(2^{O(p^2)}\,n^{O(1)}\) time, fixed-parameter tractable in \(p\), settling the cogirth side of a question of Geelen and Kapadia \citep{geelen_computing_2018}.\\
The extensions of the contraction method carry over deterministically: enumerating all near-minimum \(1\)-cocycles, computing a minimum-weight \(k\)-cocycle, and computing the Pareto frontier under several positive criteria.
\end{abstract}

\section{Introduction}\label{introduction}

Karger introduced the randomized contraction algorithm for the minimum cut problem on graphs: contract a uniformly random edge, and repeat until the minimum cut can be read off directly \citep{karger_global_1993}. The analysis rests on a hitting bound. A fixed minimum cut of size \(\lambda\) contains a uniformly random edge with probability at most \(\lambda/|E| \le 2/|V|\), since every vertex has degree at least \(\lambda\). The same bound applies after each contraction, so a fixed minimum cut survives the whole process with inverse polynomial probability.

The hitting bound reads naturally in matroid language. Let \(M\) be a matroid of rank \(r\) on \(n\) elements, and let the cogirth \(\lambda(M)\) denote the minimum weight of a cocircuit of \(M\). A uniformly random element lies in a fixed minimum cocircuit with probability \(\lambda(M)/n\). The \emph{cogirth-density ratio} of \(M\) is \(\lambda(M)\,r/n\). For a connected graph the ratio is at most \(2-2/|V|\), which recovers Karger's bound. More generally, if the ratio stays bounded by \(c\) along contractions, random contraction computes the cogirth with success probability \(n^{-O(c)}\). Karger generalizes the contraction analysis to matroids, bounding the hitting probability through the number of disjoint bases \citep[Lemma 4.5]{Karger98}, that is, a bound on the cogirth-strength ratio. The cogirth-strength and cogirth-density ratios are bounded together for contraction-closed classes of matroids (\cref{sec:algorithm}).

For graphs, the contraction method has been derandomized through tree packing. Karger packs spanning trees so that a minimum cut shares at most two edges with some packed tree, and finds the best such cut deterministically for each tree \citep{karger_minimum_2000}. Thorup extends the approach to minimum \(k\)-cut through recursive greedy tree packings, showing that a minimum \(k\)-cut shares at most \(2k-2\) edges with some packed tree \citep{thorup_minimum_2008}. Chekuri et al.~derive the packing from the dual of the Naor--Rabani linear program, simplify the analysis, and improve the bound to \(2k-3\) \citep{chekuri_lp_2020}.

The contraction argument shows the number of \(\alpha\)-approximate minimum cuts is \(O(n^{2\alpha})\) \citep{karger_global_1993}. This counting underlies the multicriteria theory of cuts. Armon and Zwick enumerate all Pareto-optimal cuts under several criteria in pseudo-polynomial time \citep{armon_multicriteria_2006}, Aissi et al.~prove strongly polynomial bounds for parametric and bicriteria cuts \citep{aissi_strongly_2015}, Karger enumerates all parametric minimum cuts by random interleaving \citep{karger_enumerating_2016}, and Beideman et al.~bound the number of multiobjective minimum cuts in hypergraphs \citep{beideman_multicriteria_2020}.

Geelen and Kapadia carry the contraction method to a matroid class close to graphic. A rank-\(p\) \emph{perturbed graphic matroid} (PGM) is a binary matroid represented over \(\FF_2\) by \(A(G)+P\), where \(A(G)\) is the vertex--edge incidence matrix of a multigraph \(G\) and \(P\) is a matrix of rank at most \(p\). They prove that contraction minors of sufficiently high rank contain small cocircuits \citep[Lemma 3.6]{geelen_computing_2018}, and they give a randomized contraction algorithm that computes the cogirth of a rank-\(p\) perturbed graphic matroid of rank \(r\) with \(n\) elements in time \(2^{2^p+O(p)}\,r^5n\), where the parameter factor follows from the rank-\(2^p\) base case of their analysis.

Known deterministic results for perturbed graphic matroids are partial. Nägele et al.~solve a submodular generalization of the \(p\)-set even-cut problem, a special case of the cogirth of perturbed graphic matroids asking for a minimum cut whose side meets each of \(p\) given vertex sets evenly, in \(n^{O(p)}\) time \citep{nagele_submodular_2019}. The algorithm of Fomin et al.~is fixed-parameter tractable only in \(p\) and \(\lambda(M)\) jointly \citep{fomin_covering_2019}. Beyond bounded perturbation rank the problem is hard. The cogirth of a general binary matroid is the minimum distance of a linear code, which is NP-hard to compute \citep{Vardy97} and W{[}1{]}-hard parameterized by the distance \citep{BhattacharyyaEtAl21}. Geelen and Kapadia ask for efficient deterministic algorithms for the girth and the cogirth of perturbed graphic matroids, remarking that finding them ``seems to be quite difficult'' \citep{geelen_computing_2018}.

\textbf{Our results.} Let \(M\) be a loopless matroid of rank \(r\) on \(n\) elements with weights \(w\), given by an independence oracle, and suppose every contraction minor \(N\) of rank at least \(r_0\) has cogirth-density ratio \(\lambda(N)\,r(N)/w(E(N))\) at most \(c\), the \emph{density hypothesis}. Write \(m\) for the maximum number of parallel classes of a contraction minor of bounded rank. Deterministic algorithms then match the randomized exponents.

\begin{itemize}
\tightlist
\item
  A minimum-weight cocircuit of \(M\) is computable in \(m^{O(r_0)}\,n^{O(c)}\) time. The algorithm knows neither \(r_0\) nor \(c\).
\item
  The cogirth of a weighted rank-\(p\) perturbed graphic matroid is computable deterministically in \(2^{O(p^2)}\,n^{O(1)}\) time, fixed-parameter tractable in \(p\). This settles the cogirth side of the question of Geelen and Kapadia.
\item
  For every \(k\le r\), a minimum-weight \(k\)-cocycle, a minimum-weight set whose removal drops the rank by \(k\), is computable in \(m^{O(r_0)}\,n^{O(kc)}\) time.
\item
  For every \(\theta\ge1\), when \(\lambda(M)>0\), a list of \(1\)-cocycles containing every \(1\)-cocycle of weight at most \(\theta\lambda(M)\) is computable in \(m^{O(r_0)}\,n^{O(\theta c)}\) time.
\item
  If the density hypothesis holds under every positive weight function, then for fixed \(q\) all Pareto-optimal cocircuits under \(q\) positive criteria are computable in \(m^{O(r_0)}\,n^{O(qc)}\) time, up to a factor polynomial in the logarithm of the weight spread.
\end{itemize}

The cogirth algorithm and the extensions are deterministic counterparts, at the matroid level, of the contraction-method bounds above, previously available deterministically only for graphs and constant-rank hypergraphs. The weight of the paper is in the reduction of minimum-weight \(k\)-cocycles to the cogirth of a truncation while maintaining the density bound, and in the perturbed graphic matroid theorem, singly exponential in \(p\) where the randomized parameter factor \(2^{2^p+O(p)}\) is doubly exponential.

\begin{table}[!htbp]
\centering
\caption{Computing the cogirth of a rank-\(p\) perturbed graphic matroid of rank
\(r\), where \(\lambda\) is the cogirth and \(g\) is a computable
function.}\label{tbl:pgm}
\small
\begin{tabularx}{\textwidth}{@{}XlXX@{}}
\toprule
Reference & Type & Running time & Scope \\
\midrule
Geelen--Kapadia & randomized & \(2^{2^p+O(p)}\,r^5n\) & PGM \\
Nägele et al. & deterministic & \(n^{O(p)}\) & \(p\)-set even-cut \\
Fomin et al. & deterministic & \(g(p,\lambda)\,n^{O(1)}\) & PGM \\
this paper & deterministic & \(2^{O(p^2)}\,n^{O(1)}\) & PGM \\
\bottomrule
\end{tabularx}
\end{table}

The method is the matroid form of tree packing: we pack bases of the \emph{truncation} \(T_{r_0-1}(M)\), the matroid whose independent sets are the independent sets of \(M\) of size at most \(r-r_0+1\), so that some packed base shares at most \(\lfloor c\rfloor\) elements with a minimum-weight cocircuit and contracting the rest of that base leaves a bounded-rank residual. The main algorithm certifies optimality from the packing value without knowing \(r_0\) or \(c\).

The running times cannot avoid \(k\) in the exponent, nor a superpolynomial factor in \(p\). At \(p=0\) the minimum-weight \(k\)-cocycle problem contains minimum \((k+1)\)-cut in graphs, since a minimum-weight \(k\)-cocycle of the graphic matroid of a connected graph is a minimum-weight edge set whose removal leaves at least \(k+1\) components. Minimum \(k\)-cut is W{[}1{]}-hard parameterized by \(k\) \citep{DowneyEFPR03} and, assuming ETH, admits no \(n^{o(k)}\)-time algorithm \citep{gupta_optimal_2022}, so the linear dependence on \(k\) in the exponent is asymptotically tight already for graphic matroids. The dependence on \(p\) cannot be polynomial either: every binary matroid of rank \(\rho\) is a rank-\((\rho+1)\) PGM, by viewing all elements as parallel copies of a single edge and absorbing the representation into the perturbation, so computing the cogirth of PGMs with unbounded \(p\) is the minimum-distance problem and carries its hardness.

\section{Preliminaries}\label{sec:prelim}

Let \(M\) be a matroid on ground set \(E\) with rank function \(r\). We write \(r(M)\) for \(r(E)\), \(E(M)\) for the ground set, and \(n:=|E|\). The closure of \(X\subseteq E\) is \(\operatorname{cl}(X):=\{e\in E:r(X\cup\{e\})=r(X)\}\), and we subscript \(r_M\) and \(\operatorname{cl}_M\) only when several matroids are in play. A \textbf{flat} is a set \(F\subseteq E\) with \(r(F\cup\{e\})>r(F)\) for all \(e\in E\setminus F\), a \textbf{hyperplane} is a flat of rank \(r(M)-1\), and a \textbf{cocycle} is the complement of a flat. A \textbf{\(k\)-cocycle} is a cocycle \(S\) with \(r(E\setminus S)\le r(M)-k\), and a \textbf{cocircuit} is an inclusion-minimal \(1\)-cocycle \citep{Oxley06}.
If removing \(X\) drops the rank by \(k\), then \(E\setminus\operatorname{cl}(E\setminus X)\subseteq X\) is a \(k\)-cocycle, so the minimum weight of a set whose removal drops the rank by \(k\) is attained by a \(k\)-cocycle.
Every matroid carries a weight function \(w:E\to\mathbb R_{\ge0}\), and all quantities are weighted. Put \(\lambda_k(M):=\min\{w(S):S\text{ a }k\text{-cocycle}\}\). The \textbf{cogirth} is \(\lambda:=\lambda_1\), the minimum weight of a cocircuit, since every \(1\)-cocycle contains a cocircuit and weights are nonnegative. A \textbf{fractional base packing} assigns coefficients \(\alpha_i\ge0\) to bases \(B_i\) with load \(\sum_{i:e\in B_i}\alpha_i\le w_e\) on each element \(e\), and its \textbf{value} is \(\sum_i\alpha_i\). The \textbf{strength} \citep{cunningham_optimal_1985} is \(\sigma(M):=\min\{w(E\setminus F)/(r(M)-r(F)):F\text{ a flat with }r(F)<r(M)\}\), equal to the maximum value of a fractional base packing \citep{cunningham_testing_1984}. The \textbf{density} of a positive-rank \(M\) is \(w(E(M))/r(M)\), and its \textbf{cogirth-density ratio} is \(\lambda(M)\,r(M)/w(E(M))\), read in product form: ratio at most \(c\) means \(\lambda(M)\,r(M)\le c\,w(E(M))\). We work with loopless matroids of positive rank. Loops lie in every flat, so they change no \(\lambda_k\) and no strength. We delete them without comment, and density and the cogirth-density ratio always refer to the loopless matroid. The \textbf{contraction} \(M/Y\) is the matroid on \(E\setminus Y\) with rank function \(X\mapsto r(X\cup Y)-r(Y)\), a \textbf{contraction minor} is any \(M/Y\), and both carry the restriction of \(w\). The cocircuits of \(M/Y\) are the cocircuits of \(M\) disjoint from \(Y\) \citep{Oxley06}. For \(r_0\ge1\) and \(c\ge1\), a matroid is \textbf{\((r_0,c)\)-dense} if every contraction minor of rank at least \(r_0\) has cogirth-density ratio at most \(c\). Two nonloop elements \(e\) and \(e'\) are \textbf{parallel} if \(r(\{e,e'\})=1\), a \textbf{parallel class} is a maximal set of pairwise parallel elements, and \(\varepsilon(M)\) denotes the number of parallel classes of \(M\).

For \(0\le h<r(M)\), the \textbf{rank-\((r(M)-h)\) truncation} \(T_h(M)\) is the matroid on \(E\) whose independent sets are the independent sets of \(M\) of size at most \(r(M)-h\), with rank function \(X\mapsto\min(r(X),\,r(M)-h)\).

\section{Packing truncations}\label{sec:algorithm}

\subsection{Truncation packing}\label{truncation-packing}

\begin{algorithm}[H]\caption{Truncation packing.}\label{alg:truncation-packing}

\begin{coflatalgo}
\(\textsc{TruncationPacking}(M, w, \tau, h)\):
\+\\ compute a basic optimal fractional base packing
\(\alpha_1B_1,\dots,\alpha_tB_t\) of \(T_h(M)\)
\\ for every \(B_i\) and every \(Z\subseteq B_i\) with \(|Z|\le\tau\)
\+\\ if the residual \(M/(B_i\setminus Z)\) has positive rank
\+\\ compute its minimum-weight cocircuit
\-\-\\ return the lightest cocircuit found
\end{coflatalgo}

\end{algorithm}

A run of \cref{alg:truncation-packing} takes two parameters beyond the input. The \textbf{depth} \(h\) selects the truncation: the packed sets \(B_i\) are bases of \(T_h(M)\), independent sets of size \(r-h\). The \textbf{cap} \(\tau\) limits the retained set \(Z\), and a \textbf{branch} contracts \(B_i\setminus Z\) for one packed base \(B_i\) and one retained set \(Z\subseteq B_i\) with \(|Z|\le\tau\). The \textbf{trace} of a packed base \(B_i\) on a set \(A\subseteq E(M)\) is \(B_i\cap A\). A basic packing has at most \(n\) bases, so a run enumerates at most \((\tau+1)\,n^{1+\tau}\) branches.

Packing bases of \(M\) itself would cap the smallest trace only by \(\lambda(M)/\sigma(M)\), and \((r_0,c)\)-density with \(r_0>1\) places no bound on that ratio, since the strength sees contraction minors of rank below \(r_0\). The direct sum \(M_0\) of the graphic matroid of a complete graph on more than \(2^{p-1}\) vertices and the binary projective geometry \(PG(p-1,2)\), with unit weights, is a rank-\(p\) PGM, hence \((6(2^p+p+1),\,24/5)\)-dense (\cref{lem:pgm-density} below), while the flat spanned by the clique gives \(\lambda(M_0)/\sigma(M_0)>p/2\).

\begin{lemma}\label{lem:truncation-strength}

Let \(M\) be loopless of rank \(r\) with weights \(w\), let \(c>0\), and let \(0\le h<r\). If every contraction minor \(M/Y\) with \(r(M/Y)\ge h+1\) satisfies \(\lambda(M/Y)\le c\,w(E(M/Y))/r(M/Y)\), then \(\sigma(T_h(M))\ge\lambda(M)/c\).\end{lemma}

\begin{proof}

Let \(F\) be a flat of \(T_h(M)\) with \(r_{T_h(M)}(F)<r-h\). The rank functions of \(M\) and \(T_h(M)\) agree below \(r-h\), so \(r_M(F)=r_{T_h(M)}(F)\), the set \(F\) is a flat of \(M\), and \(N:=M/F\) is loopless with \(E(N)=E\setminus F\) and \(r(N)=r-r_M(F)\ge h+1\ge1\). Cocircuits of \(N\) are cocircuits of \(M\) (\cref{sec:prelim}), so \(\lambda(M)\le\lambda(N)\). The hypothesis at \(N\) gives \(\lambda(M)\le c\,w(E\setminus F)/r(N)\le c\,w(E\setminus F)/\bigl((r-h)-r_{T_h(M)}(F)\bigr)\), using \(r(N)=r-r_M(F)\ge(r-h)-r_{T_h(M)}(F)\). Minimizing over \(F\) proves the claim.\end{proof}

At \(h=0\) the lemma gives \(\sigma(M)\ge\lambda(M)/c\) for a \((1,c)\)-dense matroid, and \(\sigma(M)\le w(E(M))/r(M)\) through the flat \(\operatorname{cl}(\varnothing)\), so on contraction-closed hypotheses the cogirth-strength ratio \(\lambda(M)/\sigma(M)\) of Karger \citep{Karger98} and the cogirth-density ratio are bounded by the same constants.

With \cref{lem:truncation-strength} in place, the proof of the theorem below is the tree-packing scheme for minimum cut \citep{karger_minimum_2000, thorup_minimum_2008}, with the packing value entering through the averaging computation of Chekuri et al. \citep[Section 4]{chekuri_lp_2020}.

\begin{theorem}\label{thm:main}

Let \(M\) be a loopless \((r_0,c)\)-dense matroid of rank \(r\ge1\) on \(n\) elements, given by an independence oracle, with weights \(w:E(M)\to\mathbb R_{\ge0}\). Set \(d:=r_0-1+\lfloor c\rfloor\), and suppose a minimum-weight cocircuit of any contraction minor of \(M\) of positive rank at most \(d\) is computable in time \(\beta\). Then a minimum-weight cocircuit of \(M\) is computable deterministically in \(n^{O(c)}\,(\beta+n^{O(1)})\) time.\end{theorem}

\begin{proof}

Put \(h:=r_0-1\). If \(E_0:=\{e:w_e=0\}\) satisfies \(r(E\setminus E_0)\le r-1\), we output a cocircuit inside the weight-\(0\) \(1\)-cocycle \(E\setminus\operatorname{cl}(E\setminus E_0)\), and if \(r\le h\), we solve \(M\) itself in time \(\beta\). Otherwise \(\lambda(M)>0\), since a \(1\)-cocycle of weight \(0\) would lie in \(E_0\), and \(\sigma(T_h(M))\ge\lambda(M)/c>0\) by \((r_0,c)\)-density and \cref{lem:truncation-strength}.

We compute the packing step of \cref{alg:truncation-packing}. A basic optimal solution of the packing linear program, with a variable per base of \(T_h(M)\) and a load constraint per element, has \(t\le n\) bases and is computable in polynomial time \citep{cunningham_testing_1984, quanrud_faster_2024}, and the rank oracle of \(T_h(M)\) is immediate from that of \(M\). Let \(C^\star\) be a minimum-weight cocircuit. Drawing \(B_i\) with probability \(\alpha_i/\sigma(T_h(M))\) gives
\[
\mathbb E\,|B_i\cap C^\star|\ =\ \frac{1}{\sigma(T_h(M))}\sum_{e\in C^\star}\ \sum_{i:e\in B_i}\alpha_i\ \le\ \frac{w(C^\star)}{\sigma(T_h(M))}\ \le\ c,
\]
so some \(B_i\) has \(|B_i\cap C^\star|\le\lfloor c\rfloor\). We run \cref{alg:truncation-packing} with cap \(\tau=\lfloor c\rfloor\) and depth \(h\), for at most \((\lfloor c\rfloor+1)\,n^{1+\lfloor c\rfloor}\) branches. Each residual is a contraction minor of rank \(h+|Z|\le d\), and each residual of positive rank is solved in time \(\beta\) by the assumed base case.

Every branch returns a cocircuit of a contraction minor of \(M\), hence a cocircuit of \(M\) (\cref{sec:prelim}) of weight at least \(\lambda(M)\). In the branch with \(Z:=B_i\cap C^\star\), the contracted set \(B_i\setminus Z\) is independent and disjoint from \(C^\star\), so removing \(C^\star\) still drops the rank of the residual, which therefore has a cocircuit of weight at most \(w(C^\star)\). Each branch costs \(\beta+n^{O(1)}\).\end{proof}

\begin{corollary}\label{cor:points}

Let \(M\) be a loopless \((r_0,c)\)-dense matroid of rank \(r\ge1\) on \(n\) elements, given by an independence oracle, with weights \(w:E(M)\to\mathbb R_{\ge0}\), and set \(d:=r_0-1+\lfloor c\rfloor\). Put \(m:=\max\{\varepsilon(N):N\text{ a contraction minor of }M,\ r(N)\le d\}\). Then a minimum-weight cocircuit of \(M\) is computable deterministically in \(m^{O(r_0)}\,n^{O(c)}\) time.\end{corollary}

\begin{proof}

We instantiate the base case of \cref{thm:main}. Given a contraction minor of positive rank at most \(d\), remove loops and collapse each parallel class to one element carrying its total weight: flats are closed under parallelism, so cocircuits are unions of parallel classes and an optimum uncollapses. The collapsed minor has at most \(m\) elements and rank \(\rho\le d\), and its cocircuit complements are the closures of its independent sets of size \(\rho-1\), enumerable in \(m^{d-1+O(1)}\) oracle calls. Hence \(\beta\le m^{d-1+O(1)}+n^{O(1)}\), and \cref{thm:main} runs in \(n^{O(c)}(\beta+n^{O(1)})=m^{O(r_0)}\,n^{O(c)}\) time by \(m\le n\).\end{proof}

Since \(m\le n\) in all cases, \cref{cor:points} computes a minimum-weight cocircuit of any loopless \((r_0,c)\)-dense matroid in \(n^{O(r_0+c)}\) time. The content of \(r_0\) is the split \(m^{O(r_0)}\,n^{O(c)}\).

\subsection{Parameter-free variant}\label{parameter-free-variant}

\begin{algorithm}[H]\caption{Parameter-free truncation packing.}\label{alg:parameter-free}

\begin{coflatalgo}
\(\textsc{ParameterFree}(M, w)\):
\+\\ collapse the parallel classes of \(M\)
\\ if \(E_0:=\{e:w_e=0\}\) has \(r(E\setminus E_0)<r\)
\+\\ return a cocircuit inside
\(E\setminus\operatorname{cl}(E\setminus E_0)\)
\-\\ form the closure arm, and one packing arm for each depth
\(h\in\{0,\dots,r-2\}\)
\\ for budget \(b=1,2,4,\dots\)
\+\\ run every arm from scratch for at most \(b\) steps
\\ if some arm halted
\+\\ return its answer, uncollapsed
\\ 
\-\-\\ the closure arm:
\+\\ record \(E\setminus\operatorname{cl}(I)\) for every independent set
\(I\) of size \(r-1\)
\\ halt with the lightest recorded cocircuit
\\ 
\-\\ the packing arm at depth \(h\):
\+\\ maintain the witness, the lightest cocircuit found, of weight \(\mu\),
initially \(\infty\)
\\ for cap \(s=1,2,\dots,r-h-1\)
\+\\ run \(\textsc{TruncationPacking}(M,w,s,h)\), recursing on the residuals
\\ if \((s+1)\,\sigma(T_h(M))>\mu\)
\+\\ halt with the witness
\end{coflatalgo}

\end{algorithm}

The main algorithm \cref{alg:parameter-free} hedges over the unknown depth and cap with \(r\) independent procedures called \textbf{arms}, the closure arm and one packing arm per depth \(h\in\{0,\dots,r-2\}\), run in rounds under a doubling step budget, and a packing arm halts once the \textbf{certificate} \((s+1)\,\sigma(T_h(M))>\mu\) validates its witness. A \textbf{step} is one oracle call or one unit of computation. Collapsing and uncollapsing are as in the proof of \cref{cor:points}.

\begin{lemma}\label{lem:correctness}

Given a loopless matroid \(M\) of positive rank by an independence oracle, with weights \(w:E(M)\to\mathbb R_{\ge0}\), \cref{alg:parameter-free} returns a minimum-weight cocircuit of \(M\).\end{lemma}

\begin{proof}

The collapse changes no cocircuit weight and uncollapsing recovers an optimum, as in the proof of \cref{cor:points}. The zero-weight test settles \(\lambda(M)=0\), as the preliminary case in the proof of \cref{thm:main}, so assume \(\lambda(M)>0\). The recursion is well founded: the caps stay below \(r-h\), so every contracted set is nonempty and every recursive call has strictly smaller rank. We assume by induction on the rank that the recursive calls are exact. Every set any arm records is a cocircuit of \(M\), directly or as a cocircuit of a contraction minor (\cref{sec:prelim}), so \(\mu\ge\lambda(M)\) throughout. The closure arm always halts and its answer is exact, since every cocircuit of \(M\) is the complement of a hyperplane and every hyperplane is the closure of an independent set of size \(r-1\). Suppose a packing arm at depth \(h\) halts after completing the cap \(s\). The certificate forces \(\sigma(T_h(M))>0\), since \(\mu\ge\lambda(M)>0\). The averaging argument of \cref{thm:main} then gives a packed base whose trace on a minimum-weight cocircuit \(C^\star\) is an integer at most \(w(C^\star)/\sigma(T_h(M))\le\mu/\sigma(T_h(M))<s+1\), hence at most \(s\), so the completed run visited a branch that contracted a set disjoint from \(C^\star\). The residual of that branch has positive rank, since removing \(C^\star\) drops its rank, and the exact recursive call on it found a cocircuit of weight at most \(w(C^\star)\), so \(\mu=\lambda(M)\) and the witness is optimal. \cref{alg:parameter-free} therefore always returns a minimum-weight cocircuit.\end{proof}

\begin{proposition}\label{prop:overhead}

Let \(M\) be a loopless \((r_0,c)\)-dense matroid of rank \(r\ge1\) on \(n\) elements, given by an independence oracle, with weights \(w:E(M)\to\mathbb R_{\ge0}\), and set \(d:=r_0-1+\lfloor c\rfloor\) and \(m:=\max\{\varepsilon(N):N\text{ a contraction minor of }M,\ r(N)\le d\}\). Then \cref{alg:parameter-free} computes a minimum-weight cocircuit of \(M\) in \(r\,n^{1+\lfloor c\rfloor}\,\bigl(m^{d-1+O(1)}+n^{O(1)}\bigr)\) time.\end{proposition}

\begin{proof}

The collapse and the zero-weight test are polynomial, so assume \(\lambda(M)>0\). If \(r\le d\), then \(M\) collapses to at most \(m\) elements, the closure arm halts within \(m^{r-1+O(1)}\) steps, and the claimed bound holds. Now assume \(r>d\).
The algorithm runs \(r\) arms under doubling budgets, so its total cost is at most \(4r\) times the step count at which the first arm halts, and it suffices to bound the packing arm at depth \(r_0-1\). As in the proof of \cref{thm:main}, \(\sigma(T_{r_0-1}(M))\ge\lambda(M)/c>0\). The sweep reaches the cap \(s^*:=\max\{1,\lfloor\lambda(M)/\sigma(T_{r_0-1}(M))\rfloor\}\le\lfloor c\rfloor\), since \(r>d\) gives \(r-(r_0-1)-1\ge\lfloor c\rfloor\). Once the sweep completes the cap \(s^*\), the averaging argument has placed a branch disjoint from an optimum, exact recursion has set \(\mu=\lambda(M)\), and the certificate holds. The cost of consecutive caps grows by factors of \(n\), so the whole sweep costs at most twice the run at cap \(s^*\), which spawns at most \((\lfloor c\rfloor+1)\,n^{1+\lfloor c\rfloor}\) recursive calls on contraction minors of rank at most \((r_0-1)+s^*\le d\). A recursive call on a minor of rank \(\rho\le d\) collapses it to at most \(m\) elements, runs at most \(\rho\le m\) arms, and its closure arm halts within \(m^{\rho-1+O(1)}\) steps, so the call costs \(m^{d-1+O(1)}+n^{O(1)}\). The arm therefore costs \(n^{1+\lfloor c\rfloor}\,(m^{d-1+O(1)}+n^{O(1)})\).\end{proof}

Running the single arm at depth \(r_0-1\) with cap \(\lfloor c\rfloor\) is exactly the run analyzed in \cref{thm:main,cor:points}, so the factor \(r\) is the whole price of not knowing \(r_0\) and \(c\).

\section{Perturbed graphic matroids}\label{sec:pgm}

We write \(M(A)\) for the matroid on the columns of a matrix \(A\), and call the row space of \(A\) its \textbf{cocycle space}: the support of any nonzero cocycle-space vector is a cocycle and contains a cocircuit, and the cocycle space of a contraction \(M(A)/Y\) consists of the cocycle-space vectors vanishing on \(Y\), restricted to the remaining columns \citep{Oxley06}.

Represent the rank-\(p\) PGM by \(A(G)+P\) with \(A(G)\in\FF_2^{V(G)\times E(G)}\) and \(P\in\FF_2^{V(G)\times E(G)}\) of rank at most \(p\), and write \(P=UW\) with \(U\) having \(p\) columns. The block matrix \(\bigl(\begin{smallmatrix}A(G)&U\\W&I_p\end{smallmatrix}\bigr)\), contracted on its last \(p\) columns, represents \(A(G)+P\) by Schur-complement elimination. The minors of rank-\(p\) PGMs form a contraction-closed class, and the block form passes to them:

\begin{lemma}\label{lem:closure}

Every minor \(N\) of a rank-\(p\) PGM is \(M(B)/L\) for a binary matrix \(B=\bigl(\begin{smallmatrix}A(G)&U\\W&Q\end{smallmatrix}\bigr)\), computable from a representation of the PGM, in which \(A(G)\) is the vertex--edge incidence matrix of a graph with loops permitted as zero columns, the rows of \(A(G)\) are the \textbf{vertex rows}, the blocks \(W\) and \(Q\) occupy \(p\) appended rows, and \(L\) is the set of the \(p\) appended columns.\end{lemma}

\begin{proof}

The Schur form above expresses the PGM itself, and deleting an element deletes its column. For a contraction by an element \(e\) of \(N\), put the column of \(e\) into one of three cases. If its \(A(G)\)-part is nonzero, pivot in a vertex row. Deleting the pivot row and column replaces \(A(G)\) by the incidence matrix of the contracted graph, and elsewhere the row operations affect only the \(U\), \(W\), and \(Q\) blocks. If its \(A(G)\)-part is zero but its lower part is nonzero, pivot in an appended row and delete the pivot row and column, padding with a zero row to restore \(p\) appended rows. If the whole column is zero, then \(e\) is a loop and contraction equals deletion. Each case preserves the block form and the incidence form of the top-left block.\end{proof}

The bound of \citep[Lemma 3.6]{geelen_computing_2018} does not give \((r_0,c)\)-density directly, since it applies to connected instances of their labelled-graph encoding rather than to arbitrary contraction minors, and it measures density against the vertex count rather than the rank.

\begin{lemma}\label{lem:pgm-density}

Let \(N\) be a loopless minor of a rank-\(p\) PGM, with weights \(w\), and write \(r:=r(N)\).

\begin{enumerate}
\def\labelenumi{\arabic{enumi}.}
\tightlist
\item
  If \(r\ge 2^p+p+2\), then \(\lambda(N)\le 4\,w(E(N))/(r-2^p-p-1)\).
\item
  If \(r\ge 4p+2\), then \(\lambda(N)\le 4(p+1)\,w(E(N))/r\).
\end{enumerate}

\end{lemma}

\begin{proof}

By \cref{lem:closure}, \(N=M(B)/L\) with \(B=\bigl(\begin{smallmatrix}A(G)&U\\W&Q\end{smallmatrix}\bigr)\) and \(L\) the set of the \(p\) appended columns. Split \(B=A_0+P\) with \(A_0:=\bigl(\begin{smallmatrix}A(G)&0\\0&0\end{smallmatrix}\bigr)\): the matrix \(A_0\) has at most two nonzero entries per column and vanishes on \(L\), and the vertex rows of \(P\) are \((0\ \ U)\), supported on the \(p\) columns of \(L\).

Fix a set \(R\) of \(\operatorname{rank}(B)\ge r\) linearly independent rows of \(B\). For a row \(v\) let \(\omega(v):=w(\operatorname{supp}((A_0)_v))\), the weight of the support of the \emph{unperturbed} row. Since \(A_0\) vanishes on \(L\) and each column meets at most two rows, \(\sum_{v\in R}\omega(v)\le2\,w(E(N))\). Set \(j:=2^p\) in case 1 and \(j:=p\) in case 2, and let \(J\) be the \(j+1\) vertex rows of \(R\) of smallest \(\omega\). The set \(R\) has at least \(r-p\) vertex rows, and \(r-p\ge j+2\) holds in case 1 by \(r\ge2^p+p+2\) and in case 2 by \(r\ge2p+2\), so \(J\) exists and at least \(r-p-j-1\ge1\) vertex rows of \(R\) lie outside it. Each vertex row of \(R\) outside \(J\) has \(\omega\) at least \(\max_{v\in J}\omega(v)\), and hence \(\omega(v)\le 2\,w(E(N))/(r-p-j-1)\) for every \(v\in J\).

In case 1, the rows \(\{P_v:v\in J\}\) are \(2^p+1\) vectors supported on the \(p\) columns of \(L\), so they take at most \(2^p\) values, and by pigeonhole \(P_u=P_v\) for some two distinct \(u,v\in J\), and we set \(x:=\chi_u+\chi_v\). In case 2, the rows \(\{P_v:v\in J\}\) are \(p+1\) vectors in a space of dimension at most \(p\), so some nonzero \(x\) supported on \(J\) has \(x^\top P=0\). In both cases \(x^\top P=0\), so \(y:=x^\top B=x^\top A_0\). The vector \(y\) is a combination of the linearly independent rows of \(B\) with not all coefficients zero, so \(y\ne0\). Since \(A_0\) vanishes on \(L\), so does \(y\). Hence \(y\) restricts to a nonzero cocycle-space vector of \(N\), and \(\operatorname{supp}(y)\setminus L\) contains a cocircuit of \(N\) of weight at most \(\sum_{v\in J\cap\operatorname{supp}(x)}\omega(v)\). In case 1 this weight is at most \(\omega(u)+\omega(v)\le 4\,w(E(N))/(r-2^p-p-1)\). In case 2 it is at most \(\sum_{v\in J}\omega(v)\le 2(p+1)\,w(E(N))/(r-2p-1)\le 4(p+1)\,w(E(N))/r\), using \(r\ge4p+2\).\end{proof}

The pigeonhole of case 1 is the charging argument of Geelen and Kapadia \citep[Lemma 3.6]{geelen_computing_2018}, there phrased for even cuts of labelled graphs.

The ratio \(4r(N)/(r(N)-2^p-p-1)\) of case 1 is at most \(24/5\) once \(r(N)\ge6(2^p+p+1)\), so every minor of a rank-\(p\) PGM is \((r_0,c)\)-dense for \((r_0,c)=(6(2^p+p+1),\,24/5)\), and \cref{cor:points} with the trivial bound \(m\le n\) already computes the cogirth deterministically in \(n^{O(2^p)}\) time. The threshold sits in the exponent of \(n\), and it is not forced: bounded-rank minors of PGMs have at most \(2^{O(p)}\) parallel classes (\cref{lem:compression} below), which turns the same run into an FPT algorithm with parameter factor \(2^{O(p\,2^p)}\). The one remaining doubly exponential cost is the brute-force solve at rank \(\Theta(2^p)\) inside \cref{cor:points}. \cref{thm:pgm-single-exp} replaces that brute force by a second run of the algorithm: the residuals collapse to \(2^{O(p)}\) elements, and on so few elements case 2, whose ratio grows with \(p\), is affordable, since the ratio enters only the exponent of the element count.

\begin{lemma}\label{lem:compression}

Let \(N\) be a loopless minor of a rank-\(p\) PGM with weights \(w\), given in the block representation of \cref{lem:closure}. Then \(E(N)\) partitions into at most \(4(r(N)+p)^2\,2^p\) classes of pairwise parallel elements, not necessarily maximal, computable from the representation. Replacing each class by a single element carrying the class's total weight is a deletion followed by a reweighting that leaves the rank and the cogirth unchanged and yields again a loopless minor of a rank-\(p\) PGM.\end{lemma}

\begin{proof}

Elements whose columns of \(B=\bigl(\begin{smallmatrix}A(G)&U\\W&Q\end{smallmatrix}\bigr)\) are equal are parallel in \(M(B)\), hence parallel in \(N=M(B)/L\). A column consists of its \(A(G)\)-part, either zero or \(\chi_u+\chi_v\) for an edge \(uv\), and its \(W\)-part in \(\FF_2^p\). The number \(\nu'\) of vertices incident to an edge whose \(A(G)\)-part is nonzero satisfies \(\nu'\le 2\operatorname{rank}(A(G))\), since a component with \(a\ge2\) such vertices contributes \(a-1\ge a/2\) to the graphic rank, and \(\operatorname{rank}(A(G))\le\operatorname{rank}(B)\le r(N)+p\), since \(A(G)\) is a submatrix of \(B\) and contracting the \(p\) columns of \(L\) drops the rank by at most \(p\). So there are at most \(\bigl(1+\tbinom{\nu'}{2}\bigr)2^p\le 4(r(N)+p)^2\,2^p\) distinct columns. A union of maximal parallel classes is also a union of the classes of any finer partition into parallel elements, so the collapse-and-uncollapse argument in the proof of \cref{cor:points} applies verbatim. Deleting copies keeps \(N\) a loopless minor of a rank-\(p\) PGM, and each deleted element is parallel to a kept element, so the rank is unchanged.\end{proof}

\begin{theorem}\label{thm:pgm-single-exp}

The cogirth of a weighted rank-\(p\) PGM on \(n\) elements, given by the representation \(A(G)+P\), is computable deterministically in \(2^{O(p^2)}\,n^{O(1)}\) time.\end{theorem}

\begin{proof}

Assume \(p\ge1\), since a rank-\(0\) perturbation is a rank-\(1\) perturbation. The input \(M\) is \((r_0,c)\)-dense for \((r_0,c)=(6(2^p+p+1),\,24/5)\) by \cref{lem:closure} and case 1 of \cref{lem:pgm-density}, and we run \cref{thm:main} on it. It enumerates \(n^{O(c)}=n^{O(1)}\) branches, and each branch asks for a minimum-weight cocircuit of a contraction minor \(N\) of \(M\) of positive rank at most \(d=r_0-1+\lfloor c\rfloor=O(2^p)\). It remains to solve each such minor within the claimed parameter factor.

By \cref{lem:compression}, collapsing parallel classes replaces \(N\) by a minor \(N'\) of a rank-\(p\) PGM with the same rank and the same cogirth, on \(n'\le4(d+p)^2\,2^p=2^{O(p)}\) elements. Every contraction minor of \(N'\) is again a minor of a rank-\(p\) PGM, so case 2 of \cref{lem:pgm-density} shows \(N'\) is \((4p+2,\,4(p+1))\)-dense, and applying \cref{cor:points} to \(N'\) computes a minimum-weight cocircuit of \(N'\) in \(n'^{O(p)}=2^{O(p^2)}\) time. Uncollapsing yields a minimum-weight cocircuit of \(N\) of the same weight.

The collapse is computable from the representation in polynomial time (\cref{lem:compression}), so the base case of \cref{thm:main} runs in time \(\beta=2^{O(p^2)}\,n^{O(1)}\), and the total running time is \(2^{O(p^2)}\,n^{O(1)}\).\end{proof}

\section{Extensions}\label{sec:extensions}

\subsection{\texorpdfstring{Minimum-weight \(k\)-cocycles}{Minimum-weight k-cocycles}}\label{sec:kext}

The next two propositions are well known, and both follow by inspecting the rank function of the truncation \citep{Oxley06}.

\begin{proposition}\label{prop:trunc-cocircuit}

Let \(M\) be a matroid of rank \(r\) with weights \(w\) and let \(1\le k\le r\). The cocircuits of \(T_{k-1}(M)\) are the inclusion-minimal \(k\)-cocycles of \(M\), and \(\lambda(T_{k-1}(M))=\lambda_k(M)\).\end{proposition}

\begin{proposition}\label{prop:trunc-contract}

Let \(M\) be a matroid of rank \(r\), let \(0\le h<r\), and let \(Y\subseteq E(M)\) with \(r_M(Y)<r-h\). Then \(T_h(M)/Y=T_h(M/Y)\).\end{proposition}

In particular, in the setting of \cref{prop:trunc-cocircuit}, every contraction minor \(T_{k-1}(M)/Y\) of positive rank is \(T_{k-1}(N)\) for the contraction minor \(N:=M/Y\) of rank \(r(T_{k-1}(M)/Y)+k-1\ge k\), since positive rank forces \(r_M(Y)<r-k+1\). The truncation \(T_{k-1}(N)\) is loopless whenever \(N\) is, parallel elements of \(N\) stay parallel in it, and \(\varepsilon(T_{k-1}(N))\le\varepsilon(N)\).

A greedy contraction argument converts \((r_0,c)\)-density into a bound on \(\lambda_k\).

\begin{lemma}\label{lem:trunc-dense}

Let \(k\ge1\) be an integer and let \(M\) be a loopless \((r_0,c)\)-dense matroid of rank \(r\ge k\) with weights \(w\). Then \(T_{k-1}(M)\) is \((r_0,kc)\)-dense.\end{lemma}

\begin{proof}

Let \(T_{k-1}(M)/Y\) be a contraction minor of rank \(\rho\ge r_0\). As observed after \cref{prop:trunc-contract}, it is \(T_{k-1}(N)\) for the contraction minor \(N:=M/Y\) of rank \(r(N)=\rho+k-1\ge r_0+k-1\), and \(N\) is \((r_0,c)\)-dense (\cref{sec:prelim}). Since \(\lambda(T_{k-1}(N))=\lambda_k(N)\) by \cref{prop:trunc-cocircuit} and \(r(T_{k-1}(N))=\rho\), it suffices to show \(\lambda_k(N)\le kc\,w(E(N))/\rho\), which we do greedily.

Construct independent sets \(\varnothing=I_0\subsetneq\cdots\subsetneq I_k\) and cocircuits \(C_1,\dots,C_k\): given \(I_j\) with \(j<k\), the contraction minor \(N_j:=N/I_j\) has rank \(r(N)-j\ge r_0\), so \(\lambda(N_j)\le c\,w(E(N_j))/(r(N)-j)\le c\,w(E(N))/(r(N)-k+1)=c\,w(E(N))/\rho\). Let \(C_{j+1}\) be a minimum-weight cocircuit of \(N_j\), pick \(e_{j+1}\in C_{j+1}\), and set \(I_{j+1}:=I_j\cup\{e_{j+1}\}\), independent in \(N\).

The union \(S:=C_1\cup\cdots\cup C_k\) has \(w(S)\le\sum_{j=1}^k w(C_j)\le kc\,w(E(N))/\rho\), so it suffices to show that removing \(S\) drops the rank of \(N\) by \(k\). Let \(R:=E(N)\setminus S\). For each \(j\), the cocircuit \(C_j\) of \(N_{j-1}\) is a cocircuit of \(N\) disjoint from \(I_{j-1}\) (\cref{sec:prelim}), so \(H_j:=E(N)\setminus C_j\) is a hyperplane of \(N\) with \(R\cup I_{j-1}\subseteq H_j\) and \(e_j\notin H_j\), and hence \(e_j\notin\operatorname{cl}_N(R\cup I_{j-1})\) and \(r_N(R\cup I_j)=r_N(R\cup I_{j-1})+1\). Summing over \(j=1,\dots,k\) gives \(r_N(R)+k=r_N(R\cup I_k)\le r(N)\), and the \(k\)-cocycle inside \(S\) (\cref{sec:prelim}) gives \(\lambda_k(N)\le w(S)\).\end{proof}

\begin{corollary}\label{cor:kcocycle}

Let \(k\ge1\) be an integer and let \(M\) be a loopless \((r_0,c)\)-dense matroid of rank \(r\ge k\) on \(n\) elements, given by an independence oracle, with weights \(w:E(M)\to\mathbb R_{\ge0}\). Put \(m:=\max\{\varepsilon(N):N\text{ a contraction minor of }M,\ r(N)\le r_0+k-2+\lfloor kc\rfloor\}\). Then a minimum-weight \(k\)-cocycle of \(M\) is computable deterministically in \(m^{O(r_0)}\,n^{O(kc)}\) time.\end{corollary}

\begin{proof}

The truncation \(T:=T_{k-1}(M)\) is loopless of rank \(r-k+1\ge1\) on the same \(n\) elements, its independence oracle is immediate from that of \(M\), and it is \((r_0,kc)\)-dense by \cref{lem:trunc-dense}. Apply \cref{cor:points} to \(T\): its rank bound is \(d_T:=r_0-1+\lfloor kc\rfloor\), and, as observed after \cref{prop:trunc-contract}, every contraction minor of \(T\) of positive rank at most \(d_T\) is \(T_{k-1}(N)\) for a contraction minor \(N\) of \(M\) of rank at most \(d_T+k-1\) with \(\varepsilon(T_{k-1}(N))\le\varepsilon(N)\le m\). The run computes a minimum-weight cocircuit of \(T\) in \(m^{O(r_0)}\,n^{O(kc)}\) time, and by \cref{prop:trunc-cocircuit} that cocircuit is a minimum-weight \(k\)-cocycle of \(M\).\end{proof}

Specialized to the graphic matroid of a connected graph on vertex set \(V\), where every contraction minor has cogirth-density ratio at most \(2-2/|V|\), the corollary packs spanning forests with \(k\) components, and some packed forest has trace at most \(\lfloor 2k(1-1/|V|)\rfloor\le 2k-1\) on a minimum \((k+1)\)-cut. The constant matches the \(2(k+1)-3\) of Chekuri et al. \citep[Corollary 8]{chekuri_lp_2020}.

\subsection{Near-minimum enumeration}\label{sec:enum}

\begin{corollary}\label{cor:enum}

Let \(M\) be a loopless \((r_0,c)\)-dense matroid of rank \(r\ge1\) on \(n\) elements, given by an independence oracle, with weights \(w:E(M)\to\mathbb R_{\ge0}\) and \(\lambda(M)>0\), let \(\theta\ge1\), and put \(m:=\max\{\varepsilon(N):N\text{ a contraction minor of }M,\ r(N)\le r_0-1+\lfloor\theta c\rfloor\}\). A list of \(1\)-cocycles containing every \(1\)-cocycle of weight at most \(\theta\lambda(M)\) is computable deterministically in \(m^{O(r_0)}\,n^{O(\theta c)}\) time.\end{corollary}

\begin{proof}

We run the proof of \cref{thm:main} with the cap raised to \(\lfloor\theta c\rfloor\), retaining every candidate instead of the lightest. The averaging argument caps the trace of each such \(1\)-cocycle \(A\) at \(w(A)/\sigma(T_h(M))\le\theta c\), so some branch contracts a set \(Y\) disjoint from \(A\), and the flat \(E\setminus A\) contains \(Y\), so \(A\) is a \(1\)-cocycle of the residual. As the base case, we collapse the residual \(R\) to at most \(m\) elements, sound since \(1\)-cocycles are unions of parallel classes as in the proof of \cref{cor:points}, enumerate the closures of all its independent sets, retain \(E(R)\setminus F\) for every closure \(F\) with \(r(F)<r(R)\), and uncollapse, for \(m^{O(r_0+\theta c)}\) candidates per branch and \(m^{O(r_0)}\,n^{O(\theta c)}\) time in all by \(m\le n\).\end{proof}

The list also bounds the count: a loopless \((r_0,c)\)-dense matroid with \(\lambda(M)>0\) has at most \(m^{O(r_0)}\,n^{O(\theta c)}\) many \(1\)-cocycles of weight at most \(\theta\lambda(M)\). A count of \(n^{O(r_0+\theta c)}\) follows already from the survival probability in the randomized analysis, as in Karger's bounds on near-minimum cuts and matroid quotients \citep[Lemma 4.5]{karger_global_1993, Karger98}, and the contribution of the corollary is the deterministic list. For graphs, Nagamochi et al.~list all near-minimum cuts deterministically \citep{nagamochi_computing_1997}, and the corollary is an alternative route, through packing, in the matroid setting.

\subsection{Multicriteria}\label{sec:pareto}

The multicriteria minimum-cut literature runs on approximate-cut counting through contraction \citep{armon_multicriteria_2006, aissi_strongly_2015, beideman_multicriteria_2020}, and the route transfers through \cref{cor:enum}. Given criteria \(w_1,\dots,w_q:E(M)\to\mathbb R_{>0}\), a cocircuit \textbf{dominates} another if it is at most as heavy under every criterion and strictly lighter under one, and a cocircuit is \textbf{Pareto-optimal} if no cocircuit dominates it.

\begin{corollary}\label{cor:pareto}

Fix an integer \(q\ge1\), fix \(r_0\ge1\) and \(c\ge1\), and let \(M\) be a loopless matroid of positive rank on \(n\) elements, given by an independence oracle, with criteria \(w_1,\dots,w_q:E(M)\to\mathbb R_{>0}\), and put \(m:=\max\{\varepsilon(N):N\text{ a contraction minor of }M,\ r(N)\le r_0-1+\lfloor2qc\rfloor\}\). If every contraction minor of \(M\) of rank at least \(r_0\) has cogirth-density ratio at most \(c\) under every positive weight function, then all Pareto-optimal cocircuits of \(M\) are computable deterministically in \(m^{O(r_0)}\,n^{O(qc)}\) time, up to a factor polynomial in the logarithm of the weight spread \(\max_{i,e}w_i(e)/\min_{i,e}w_i(e)\).\end{corollary}

\begin{proof}

A Pareto-optimal cocircuit \(A\) satisfies \(w_\gamma(A)=q\le q\,w_\gamma(D)\) for every cocircuit \(D\) under \(w_\gamma:=\sum_i\gamma_iw_i\) with \(\gamma_i:=1/w_i(A)\), since some \(i\) has \(w_i(D)\ge w_i(A)\). Rounding \(\gamma\) up coordinatewise to the net \(\Lambda\) of combinations whose coordinates are powers of two between \(1/(n\max_{i,e}w_i(e))\) and \(2/\min_{i,e}w_i(e)\), of size \(O(\log(n\max_{i,e}w_i(e)/\min_{i,e}w_i(e)))^q\), makes \(A\) a \(2q\)-approximate minimum-weight cocircuit under some net combination. We run \cref{cor:enum} at \(\theta=2q\) under every combination in \(\Lambda\), replace every listed \(1\)-cocycle by a cocircuit inside it, and filter the union of the lists by pairwise comparison. Every Pareto-optimal cocircuit is listed, and every other cocircuit is dominated by a Pareto-optimal one, since dominance is a strict partial order on a finite set, so the filter outputs exactly the Pareto-optimal cocircuits.\end{proof}

Restricting to cocircuits loses nothing: a proper \(1\)-cocycle inside another is strictly lighter under every positive criterion, so every Pareto-optimal \(1\)-cocycle, with dominance read over \(1\)-cocycles, is inclusion-minimal and \cref{cor:pareto} lists the whole Pareto frontier over \(1\)-cocycles. The list also settles budgeted optimization, minimizing \(w_1\) over cocircuits subject to budgets on \(w_2,\dots,w_q\): a cocircuit dominating a within-budget cocircuit is itself within budget and no heavier under \(w_1\), so some Pareto-optimal cocircuit in the list attains the budgeted optimum. Parametric optima, minimum-weight cocircuits under a fixed positive combination of the criteria, are also in the list, since a dominating cocircuit would be strictly lighter under the combination, and for graphs their enumeration is Karger's, at \(O(n^{q+1})\) many \citep{karger_enumerating_2016}. Specialized to the graphic matroid of a connected graph, where every contraction minor has cogirth-density ratio at most \(2\) under every weight function, \cref{cor:pareto} enumerates all Pareto-optimal cuts deterministically in \(n^{O(q)}\) time, up to a factor polynomial in the logarithm of the weight spread. The contraction-based enumeration of Armon and Zwick is pseudo-polynomial \citep{armon_multicriteria_2006}, and a deterministic polynomial-time enumeration for constant-rank hypergraphs follows from an approach of Zenklusen reported in \citep{beideman_multicriteria_2020}, so the content of the specialization is the matroid route, and of the corollary the matroid setting.

\bibliographystyle{plainnat}
\bibliography{ref}

\end{document}